\documentclass[acmsmall,nonacm]{acmart}

\usepackage{libertine}

\usepackage{datetime}
\usepackage{url}
\usepackage{paralist}

\usepackage[font=small,labelfont=bf,justification=centering]{caption}

\usepackage{stmaryrd}

\usepackage{amsmath,amsthm}
\newtheorem{theorem}{Theorem}

\theoremstyle{definition}
\newtheorem{lemma}{Lemma}

\theoremstyle{definition}

\theoremstyle{definition}

\theoremstyle{definition}

\theoremstyle{definition}

\theoremstyle{definition}

\usepackage{cleveref}
\crefname{section}{\S}{\S\S}
\Crefname{section}{\S}{\S\S}
\crefformat{section}{\S#2#1#3}
\Crefname{line}{Line}{line}
\crefname{line}{Line}{line}
\Crefname{assumption}{Assumption}{assumption}

\crefname{table}{Table}{Tables}

\usepackage{graphicx}
\usepackage{subfig}
\usepackage{arydshln}

\usepackage{multirow}

\usepackage{threeparttable}

\usepackage{algorithmicx}
\usepackage{algorithm} 
\usepackage{algpseudocode}
\usepackage{xspace}

\newcommand{\name}{\textsc{strong}\xspace}
\newcommand{\cert}{\textsc{bucket}\xspace}

\usepackage{xcolor}
\usepackage{framed}
\definecolor{lightgray}{gray}{0.90}
\renewenvironment{leftbar}[1][\hsize]
{%
\MakeFramed{\hsize#1\advance\hsize-\width\FrameRestore}%
}
{\endMakeFramed}

\algsetblockdefx[UponSend]{UponSend}{EndUponSend}{}{}[1]{\textbf{upon} \textsc{send}($#1$)}{\phantom}
\algtext*{EndUponSend}

\algsetblockdefx[UponDeliver]{UponDeliver}{EndUponDeliver}{}{}[1]{\textbf{upon} \textsc{SecureDeliver}($#1$)}{\phantom}
\algtext*{EndUponDeliver}

\algsetblockdefx[UponInit]{UponInit}{EndUponInit}{}{}[0]{\textbf{upon} \textsc{init}}{\phantom}
\algtext*{EndUponInit}

\algsetblockdefx[UponExists]{UponExists}{EndUponExists}{}{}[2]{\textbf{upon exists} $#1$ such that $#2$}{\phantom}
\algtext*{EndUponExists}

\algsetblockdefx[Procedure]{Procedure}{EndProcedure}{}{}[2]{\textbf{procedure} {#1}($#2$)}{\phantom}
\algtext*{EndProcedure}
\algsetblockdefx[Function]{Function}{EndFunction}{}{}[2]{\textbf{function} {#1}($#2$)}{\phantom}
\algtext*{EndFunction}

\algnewcommand{\BlueComment}[1]{\textcolor{blue}{\hfill\(\triangleright\) #1}}
\algnewcommand{\LineComment}[1]{\State \(\triangleright\) #1}

\usepackage{listings}
\crefname{lstlisting}{listing}{listings}
\Crefname{lstlisting}{Listing}{Listings}

\crefname{code}{line}{lines}
\Crefname{code}{Line}{Lines}

\usepackage{xcolor}

\definecolor{mygreen}{rgb}{0.254,0.572,0.294}
\definecolor{mygray}{rgb}{0.5,0.5,0.5}
\definecolor{myorange}{rgb}{1,0.35,0}
\definecolor{mymauve}{rgb}{0.58,0,0.82}
\definecolor{myblue}{rgb}{0.2,0.4,0.6}

\definecolor{rakos4orange}{RGB}{255,165,0}
\definecolor{rakos4blue}{RGB}{14,48,173}
\definecolor{rakos4lblue}{RGB}{92,172,238}
\definecolor{rakos4dgray}{RGB}{77,77,77}
\definecolor{plainred}{RGB}{211,63,63}
\definecolor{plainorange}{RGB}{221,105,41}

\usepackage{bold-extra} 

\lstdefinelanguage{Golang}%
  {morekeywords=[1]{package,import,struct,defer,panic,%
     recover,select,var,const,iota,},%
   morekeywords=[2]{string,uint,uint8,uint16,uint32,uint64,int,int8,int16,%
     int32,int64,bool,float32,float64,complex64,complex128,byte,rune,uintptr,%
     error,interface,message,node},%
   morekeywords=[3]{map,slice,make,new,nil,len,cap,copy,close,true,false,%
     delete,append,real,imag,complex,chan,},%
   morekeywords=[4]{break,continue,goto,switch,case,fallthrough,%
    default,},%
   morekeywords=[5]{Println,Printf,Error,Send},%
   sensitive=true,%
   morecomment=[l]{//},%
   morecomment=[s]{/*}{*/},%
   morestring=[b]",%
   morestring=[s]{`}{`},%
   }

\usepackage{algorithmicx}
\usepackage{algorithm} 
\usepackage{algpseudocode}
\usepackage{xifthen}
\usepackage{graphicx}
\usepackage{wrapfig}
\usepackage{dblfloatfix}

\usepackage{stmaryrd}
\usepackage{amsmath}
\usepackage{amsthm}

\newcommand{\remove}[1]{}

\usepackage{cancel}
\newif\ifcomments
\commentstrue

\newcommand{\ms}[1]{%
	    \relax\ifmmode
	        \mathord{\mathcode`\-="702D\it #1\mathcode`\-="2200}%
	    \else
	        {\it #1}%
	    \fi
}
\newcommand{\tup}[1]{%
	    \relax\ifmmode
	      \langle #1 \rangle%
	    \else
	        $\langle$ #1 $\rangle$%
	    \fi
}

\date{}

\setcopyright{none} 
\acmConference[Anonymous Submission to ACM CCS 2021]{ACM Conference on Computer and Communications Security}{Due 06 May 2021}{London, TBD}
\acmYear{2019}

\begin{document}

\algsetblockdefx[LocalState]{LocalState}{EndLocalState}{}{}{\textbf{Local state $\alpha_p$:}}{}
\algsetblockdefx[Round]{Round}{EndRound}{}{}[1]{\textbf{Round:} $#1$}{}
\algsetblockdefx[SendStep]{SendStep}{EndSendStep}{}{}{$S_p^r$:}{}
\algsetblockdefx[TransitionStep]{TransitionStep}{EndTransitionStep}{}{}{$T_p^r$:}{}
\algsetblockdefx[Upon]{Upon}{EndUpon}{}{}[2]{\textbf{upon} $\mathit{deliver}(#1)$ from $#2$ \textbf{do}}{}
\algsetblockdefx[UponEvent]{UponEvent}{EndUponEvent}{}{}[1]{\textbf{upon event} $#1$ \textbf{do}}{}
\algsetblockdefx[UponReceipt]{UponReceipt}{EndUponReceipt}{}{}[2]{\textbf{upon receipt of} $#1$ from $#2$ \textbf{do}}{}

\author{Pierre Civit}
\affiliation{
\institution{Sorbonne University} 
\country{France}
}

\author{Seth Gilbert}
\affiliation{
\institution{NUS Singapore}
\country{Singapore}
}

\author{Rachid Guerraoui}
\affiliation{
\institution{École Polytechnique Fédérale de Lausanne (EPFL)}
\country{Switzerland}
}

\author{Jovan Komatovic}
\affiliation{
\institution{École Polytechnique Fédérale de Lausanne (EPFL)}
\country{Switzerland}
}

\author{Manuel Vidigueira}
\affiliation{
\institution{École Polytechnique Fédérale de Lausanne (EPFL)}
\country{Switzerland}
}

\title{\textbf{Strong Byzantine Agreement with Adaptive Word Complexity}
}

\maketitle

The strong Byzantine agreement (SBA) problem is defined among $n$ processes, out of which $t < n$ can be faulty and behave arbitrarily.
SBA allows correct (non-faulty) processes to agree on a common value.
Moreover, if all correct processes have proposed the same value, only that value can be agreed upon.
It has been known for a long time that any solution to the SBA problem incurs quadratic worst-case word complexity; additionally, the bound was known to be tight.
However, no existing protocol achieves \emph{adaptive} word complexity, where the number of exchanged words depends on the \emph{actual} number of faults, and not on the upper bound.
Therefore, it is still unknown whether SBA with adaptive word complexity exists.

This paper answers the question in the affirmative.
Namely, we introduce \name, a synchronous protocol that solves SBA among $n = \big(2 + \Omega(1)\big)t + 1$ processes and achieves adaptive word complexity.
We show that the fundamental challenge of adaptive SBA lies in efficiently solving \emph{certification}, the problem of obtaining a constant-sized, locally-verifiable proof that a value can safely be decided.


\section{Introduction}

Strong Byzantine agreement (SBA) is a core primitive of distributed computing.
It is indispensable to state machine replication (SMR)~\cite{CL02,Abd-El-MalekGGRW05,KotlaADCW07,Momose021}, blockchain systems~\cite{DBLP:conf/opodis/AbrahamMN0S17,buchman2016tendermint,correia2019byzantine,CGL18,DBLP:conf/sosp/GiladHMVZ17}, and various other distributed protocols~\cite{Ben-OrGW19,galil1987cryptographic,DBLP:journals/dc/GilbertLS10,DBLP:journals/tse/GuerraouiS01}.
In SBA, $n$ processes propose and agree on a value, while tolerating up to $t$ arbitrary failures.
If a process exhibits an arbitrary failure, the process is said to be faulty; otherwise, the process is said to be correct.
Formally, an SBA protocol satisfies the following guarantees:
\begin{compactitem}
    \item \emph{Termination:} All correct processes eventually decide.

    \item \emph{Agreement:} No two correct processes decide different values.

    \item \emph{Strong validity:} If all correct processes propose the same value $v$, no correct process decides a value $v' \neq v$.
\end{compactitem}
Due to its importance, SBA has been meticulously studied.
For instance, the celebrated Dolev-Reischuk bound~\cite{dolev1985bounds} shows that any SBA protocol incurs quadratic worst-case word complexity.
Importantly, the bound is proven to be tight: there exist SBA protocols with quadratic worst-case word complexity for both synchronous~\cite{Momose2021,Berman1992} and partially synchronous~\cite{civit2022byzantine,lewis2022quadratic} environments.




However, one can argue that, in practice, systems usually experience many fewer faults than their maximum tolerable threshold.
Therefore, even though it is necessary to tolerate a large portion of the system failing at any point in time, it seems wasteful to constantly pay a high price as such an event seldom occurs.
For example, an adaptive protocol boasting a word complexity of $O(n \cdot f)$, where $f \leq t$ is the actual number of faults, would have a \emph{de facto} linear factor improvement in word complexity over the best non-adaptive solutions when $f$ is small ($f \in O(1)$).
Thus, the question naturally follows: does an adaptive SBA protocol exist?




In this paper, we answer this question in the affirmative for synchronous environments.
Concretely, we present \name, the first synchronous SBA protocol with a word complexity of $O(n \cdot f)$.
Notably, \name is near-optimally resilient, operating among $n = (2 + c)t  + 1$ processes, for any fixed constant $c > 0$.


\paragraph{Technical overview}

Recently, Cohen et al. \cite{cohen2022make} have presented a Byzantine agreement protocol with external validity (EBA) achieving adaptive $O(n \cdot f)$ word complexity, where external validity refers to the satisfaction of a predefined locally-computable predicate $\mathsf{validate}(\cdot)$. In this paper, we extend this line of work by introducing a novel approach to achieve strong validity, which ensures that if all correct processes propose the same value $v$, only this value can be decided.

To address this, we build upon the concept of a Proof of Exclusivity (PoE), as introduced by Rambaud \emph{et al.}~\cite{rambaud2022linear}. A PoE protocol, such as the Big Buckets PoE \cite{rambaud2022linear}, allows a designated leader (prover) to obtain a valid value $v$ and a publicly verifiable constant-sized proof, denoted $\mathsf{PoE}(v)$, that no other value $v' \neq v$ could possibly be \emph{unanimous}, i.e., the input of all honest players. 

We propose a new protocol that leverages the PoE protocol of Rambaud \emph{et al.}~\cite{rambaud2022linear} and the EBA protocol of Cohen \emph{et al.}~\cite{cohen2022make} to provide a constant-sized proof that some value is safe to be decided. Specifically, our protocol functions in the following manner:
\begin{itemize}
    \item Suppose that each process obtains a PoE certificate $\Sigma$ that its proposal $v$ is safe such that $\mathsf{validate}(v, \Sigma) = \mathit{true}$ if and only if $\Sigma = \mathsf{PoE}(v)$.
    \item Proposing value-certificate pairs to the adaptive EBA protocol thus allows us to solve Synchronous Byzantine Agreement (SBA) with adaptive word complexity. Termination and agreement follow from the EBA protocol, whereas strong validity is ensured by the means of PoE-based certificates.
\end{itemize}

Hence, this observation reduces the adaptive SBA problem to \emph{certification}, a problem of efficiently obtaining constant-sized certificates vouching for safe values.

To solve the certification problem, we propose \cert, a protocol which exchanges $O(n \cdot f)$ words and produces certificates of $O(1)$ words assuming a threshold signature scheme, under the condition that $f \leq c \cdot t$.
Any execution of \cert unfolds in phases, where each phase has a unique leader. These phases are modeled closely on the PoE protocol outlined by Rambaud \emph{et al.}~\cite{rambaud2022linear}.
When a phase starts, all correct processes send their values to the leader.
For simplicity, let us assume that the leader receives (at least) $2t + 1$ messages.
(The case in which the leader does not receive $2t + 1$ messages is slightly more complex, and thus delegated to the later stages of the paper.)
There are two possible scenarios:
\begin{compactitem}
    \item The leader receives (at least) $t + 1$ identical values $v$.
    In this scenario, the leader forms a $(t + 1)$-combined threshold signature (containing $O(1)$ words) which proves that $v$ is a safe value.
    We say that the leader creates a \emph{positive certificate}, as the constructed certificate proves that a specific value (in this case, $v$) has indeed been proposed by a correct process, thus making the value safe.

    \item The leader does not receive $t + 1$ identical values.
    In this scenario, the leader cannot form a positive certificate.
    Instead, the leader needs to construct a \emph{negative certificate}, proving that not all correct processes have proposed the same value.
    Clearly, the leader could do so by including all the received messages in a certificate; however, this certificate would not be constant-sized.
    How can we construct a constant-sized negative certificate?

    When the leader receives (at least) $2t + 1$ proposals which do not allow a creation of a positive certificate, the leader separates the received proposals into \emph{groups} such that:
    \begin{compactitem}
        \item Each group contains at most $t$ proposals.

        \item For any two groups, at least $t + 1$ proposals fall collectively into those two groups.
    \end{compactitem}
    Importantly, the leader can always group the received proposals into $O(1)$ groups which abide by the aforementioned rules.
    
    In the next step, the leader disseminates the formed groups to all processes.
    When a correct process receives a group from the leader, it checks if its proposal belongs to the group.
    If it \emph{does not}, the process sends a partial signature of the group back to the leader.

    Finally, when the leader collects $t + 1$ partial signatures of a formed group, it combines the signatures into a $(t + 1)$-combined threshold signature which proves that at least one correct process has not proposed a value which belongs to this group.
    Thus, if the leader obtains such a threshold signature for \emph{every} formed group, the leader obtains a constant-sized (as there are $O(1)$ groups) certificate which proves that not all correct processes have proposed the same value.
\end{compactitem}

In brief, we design \name as a sequential composition of (1) \cert, our adaptive certification protocol which constructs constant-sized certificates on top of PoE protocol introduced by Rambaud \emph{et al.}~\cite{rambaud2022linear}, and (2) the EBA protocol introduced by Cohen \emph{et al.}~\cite{cohen2022make} which achieves adaptive word complexity.
The composition is carefully crafted to allow the production of a PoE certificate even when $f > c \cdot t$, and to avoid an $O(n^2)$ communication cost when $f=o(n)$.

\paragraph{Roadmap.}
We overview the related work in \Cref{section:related_work}.
In \Cref{section:model}, we define the system model.
We present \name as a composition of adaptive EBA and \cert in \Cref{section:strong}, whereas we devote \Cref{section:bucket_solution} to introducing \cert.
Finally, we conclude the paper in \Cref{section:conclusion}.

\section{Related Work} \label{section:related_work}

\paragraph{Byzantine agreement.}
Byzantine agreement~\cite{LSP82} is one of the most studied problems of distributed computing.
It is known that deterministic solutions exist both in synchronous~\cite{Momose2021,abraham2017brief,Berman1992} and partially synchronous systems~\cite{YMR19,BKM19,lewis2022quadratic,civit2022byzantine,CL02}.
However, the seminal FLP impossibility result~\cite{DBLP:journals/jacm/FischerLP85} states that no deterministic protocol solves consensus in an asynchronous system even if only one process can fail (and it can fail merely by crashing).
In order to circumvent the FLP impossibility, one can rely on randomization: many asynchronous randomized agreement protocols have been devised~\cite{B83, B87,abraham2019asymptotically,canetti1993fast,cohen2020not}.
Another approach of circumventing the aforementioned impossibility consists of weakening the definition of the problem by requiring termination only in some specific cases (and not always): this approach is known as the \emph{condition-based} approach~\cite{mostefaoui2003conditions,mostefaoui2001hierarchy,mostefaoui2003using,zibin2003condition}, and it requires an agreement protocol to terminate only if the proposals of correct processes satisfy some predefined conditions.

\paragraph{Round complexity of synchronous agreement protocols.}
Dolev and Strong has proven that any deterministic solution needs to run for (at least) $t + 1$ rounds in the worst-case.\footnote{Recall that $t$ denotes the maximum number of faulty processes.}
This lower bound on round complexity has been proven to be tight~\cite{LSP82}.
Chandra and Toueg~\cite{DBLP:conf/wdag/ChandraT90} introduced the concept of \emph{early stopping} in agreement protocols: they solve reliable broadcast in $f + 2$ rounds with crash-stop failures, and $2f  + 3$ rounds with omission failures.\footnote{Recall that $f \leq t$ denotes the actual number of faulty processes.}
Keidar and Rajsbaum~\cite{DBLP:journals/ipl/KeidarR03} proved that any uniform agreement requires $f + 2$ rounds, for any $f \leq t - 1$; uniform agreement is defined in the context of non-Byzantine faults, and it ensures that faulty and correct processes do not disagree among themselves.
In~\cite{DBLP:conf/spaa/ParvedyR04}, protocols terminating in $\min(f + 2, t + 1)$ rounds are presented for uniform consensus with crash and omission failures.

\paragraph{Word complexity of Byzantine agreement.}
Dolev and Reischuk~\cite{dolev1985bounds} proved that any deterministic Byzantine agreement protocol has $\Omega(t^2)$ word complexity.
This bound is shown to be tight, both in synchronous~\cite{Berman1992,Momose2021} and partially synchronous~\cite{lewis2022quadratic,civit2022byzantine} environments.
Recently, Cohen \emph{et al.}~\cite{cohen2022make} considered \emph{adaptive} word complexity of synchronous Byzantine agreement with external validity: they presented a deterministic solution with $O(n \cdot f)$ word complexity.
Moreover, it was shown that Byzantine agreement can be solved with adaptive word complexity in partial synchrony~\cite{DBLP:journals/corr/fever,lumiere}.
This paper continues the same line of research by presenting the first synchronous SBA with adaptive word complexity.
\section{System Model} \label{section:model}

\paragraph{Processes.}
We consider a static set $\{P_1, P_2, ..., P_n\}$ of $n = 2t + 1 + \lceil c \cdot t \rceil$ processes, where $c \in (0, 1]$ is a fixed constant.
At most $t$ processes can be Byzantine: these processes can behave arbitrarily.
If a process is Byzantine, we say that the process is \emph{faulty}; otherwise, we say that the process is \emph{correct}.
We denote by $0 \leq f \leq t$ the \emph{actual} number of faulty processes.
Finally, we define the \emph{optimistic threshold} parameter $t_o = \lfloor c \cdot t \rfloor$; note that $t_o \in O(t)$, which implies $t_o \in O(n)$.

\paragraph{Communication.}
Processes communicate by exchanging messages over an authenticated point-to-point network.
The communication network is \emph{reliable:} if a correct process sends a message to a correct process, the message is eventually received.
Moreover, the network is \emph{synchronous}.
Specifically, there exists a known bound $\delta$ on message delays.
The existence of the known bound on message delays allows us to design protocols in a round-based paradigm: if a correct process sends a message to another correct process at the beginning of some round, the message is received by the end of the same round.

\paragraph{Values.}
We denote by $\mathsf{Value}$ the set of values processes can propose and decide.
For simplicity, we assume that (1) all values can be ordered, and (2) $v_{\mathit{min}}$ (resp., $v_{\mathit{max}}$) is the minimum (resp., maximum) value.
We emphasize that \name can trivially be adapted to preserve its correctness and complexity even without the aforementioned assumptions.

\paragraph{Cryptographic primitives.}
In this paper, we rely on a $(k, n)$-threshold signature scheme~\cite{Libert2016} with $k = t + 1$.
In a $(k, n)$-threshold signature scheme, each process holds a distinct private key; there exists a single public key.
Each process $P_i$ can use its private key to produce a partial signature of a message $m$ by invoking $\mathsf{ShareSign}_i^k(m)$.
A partial signature $\mathit{psignature}$ of a message $m$ produced by process $P_i$ can be verified by $\mathsf{ShareVerify}_i^k(m, \mathit{psignature})$.
Lastly, a set $S = \{\mathit{psignature}_i\}$ of partial signatures, where $|S| = k$ and, for each $\mathit{psignature}_i \in S$, $\mathit{psignature}_i = \mathsf{ShareSign}_i^k(m)$, can be combined into a \emph{single} (threshold) signature by invoking $\mathsf{Combine}^k(S)$; a threshold signature $\mathit{tsignature}$ of message $m$ can be verified with $\mathsf{CombinedVerify}^k(m, \mathit{tsignature})$.
Where appropriate, invocations of $\mathsf{ShareVerify}^k(\cdot)$ and $\mathsf{CombinedVerify}^k(\cdot)$ are implicit in our descriptions of protocols.
We denote by $\mathsf{P\_Signature}$ and $\mathsf{T\_Signature}$ a partial signature and a (combined) threshold signature, respectively. 
Lastly, we underline that the EBA protocol introduced in~\cite{cohen2022make} (which we use in a ``closed-box'' manner) internally relies on a $(k, n)$-threshold signature scheme with $k = n - t_o$.

\paragraph{Word complexity.}
Let $\mathsf{Agreement}$ be a synchronous SBA protocol and let $\mathcal{E}(\mathsf{Agreement})$ denote the set of all possible executions.
Let $\alpha \in \mathcal{E}(\mathsf{Agreement})$ be an execution of $\mathsf{Agreement}$.
A \emph{word} contains a constant number of signatures and values.
Each message contains at least a single word.
We define the word complexity of $\alpha$ as the number of words sent in messages by all correct processes in $\alpha$.
The \emph{word complexity} of $\mathsf{Agreement}$ is then defined as
\begin{equation*}
\max_{\alpha \in \mathcal{E}(\mathsf{Agreement})}\bigg\{\text{word complexity of } \alpha\bigg\}.
\end{equation*}

\section{\name} \label{section:strong}

In this section, we present \name.
We start by defining the building blocks of \name (\Cref{subsection:building_blocks}).
Then, we give \name's pseudocode (\Cref{subsection:strong_pseudocode}), and prove its correctness and complexity (\Cref{subsection:strong_proof}).

\subsection{Building Blocks} \label{subsection:building_blocks}

\subsubsection{Certification} \label{subsubsection:certification_problem_definition}
The certification problem is a problem of obtaining a locally-verifiable, constant-sized cryptographic proof -- \emph{certificate} -- that some value is safe (according to some criterion).
The formal definition of the problem follows.

Let $\mathsf{Certificate}$ denote the set of certificates.
There exists a predefined function $\mathsf{validate}: \mathsf{Value} \times \mathsf{Certificate} \to \{ \mathit{true}, \mathit{false} \}$.
Intuitively, if $\mathsf{validate}(v, \Sigma) = \mathit{true}$, for some value $v$ and some certificate $\Sigma$, $\Sigma$ proves that $v$ is a safe value.
The certification problem exposes the following interface:
\begin{compactitem}
    \item \textbf{request} $\mathsf{start}(v \in \mathsf{Value})$: a process starts certification with a value $v$.

    \item \textbf{indication} $\mathsf{acquire}(v' \in \mathsf{Value}, \Sigma' \in \mathsf{Certificate})$: a process acquires a certificate $\Sigma'$ for value $v'$.

    \item \textbf{indication} $\mathsf{stop}$: a process stops certification.
\end{compactitem}
We say that a (correct or faulty) process \emph{obtains} a certificate $\Sigma$ if and only if the process stores $\Sigma$ in its local memory.
Importantly, it is assumed that all correct processes start certification simultaneously (i.e., all correct processes invoke $\mathsf{start}(\cdot)$ at the same round).

Certification ensures the following guarantees:
\begin{compactitem}
    \item \emph{Termination:} All correct processes stop certification simultaneously (i.e., at the same round).

    \item \emph{Liveness:} If $f \leq t_o$, every correct process acquires a certificate $\Sigma'$ for some value $v'$ ($\mathsf{validate}(v', \Sigma') = \mathit{true}$).

    \item \emph{Safety:} If all correct processes start certification with the same value $v$, then no (correct or faulty) process obtains a certificate for any value different from $v$.
\end{compactitem}
The termination property states that all correct processes simultaneously terminate certification.
Liveness guarantees that, if the number of faults does not exceed the optimistic threshold, all correct processes eventually acquire a certificate for some value.
Finally, safety ensures that, if all processes start with the same value, no process (even if faulty) obtains a certificate for any other value.

\paragraph{Implementation employed by \name.}
\name employs \cert as its certification protocol.
Importantly, \cert (see \Cref{section:bucket_solution}) solves the certification problem with $O(n \cdot f)$ exchanged words.

\subsubsection{Byzantine Agreement with External Validity}
The Byzantine agreement problem with external validity (EBA) is well-studied in the literature~\cite{cohen2022make,civit2022byzantine,YMR19,lewis2022quadratic}.
EBA is similar to the SBA problem, except that processes propose (and decide) value-certificate pair, and EBA guarantees external validity instead of strong validity.
Formally, EBA exposes the following interface:
\begin{compactitem}
    \item \textbf{request} $\mathsf{propose}(v \in \mathsf{Value}, \Sigma \in \mathsf{Certificate})$: a process proposes a value-certificate pair $(v, \Sigma)$.

    \item \textbf{indication} $\mathsf{decide}(v' \in \mathsf{Value}, \Sigma' \in \mathsf{Certificate})$: a process decides a value-certificate pair $(v', \Sigma')$.
\end{compactitem}
We assume that a correct process invokes $\mathsf{propose}(v, \Sigma)$, for some value-certificate pair $(v, \Sigma)$, only if $\mathsf{validate}(v, \Sigma) = \mathit{true}$.

EBA ensures the following properties:
\begin{compactitem}
    \item \emph{Termination:} Every correct process eventually decides.
    
    \item \emph{Agreement:} No two correct processes decide different value-certificate pairs.

    \item \emph{External validity:} If a correct process decides a value-certificate pair $(v, \Sigma)$, then $\mathsf{validate}(v, \Sigma) = \mathit{true}$. 
\end{compactitem}
To make the problem non-trivial, it is assumed that correct processes do not know which value-certificate pairs are valid.
(Otherwise, all correct processes could decide a predetermined valid pair without any communication.)

\paragraph{Implementation employed by \name.}
In \name, we rely on a specific implementation of EBA proposed by Cohen \emph{et al.}~\cite{cohen2022make}.
The aforementioned implementation solves the EBA problem with adaptive word complexity of $O(n \cdot f)$.

\subsection{Protocol} \label{subsection:strong_pseudocode}

The protocol of \name is given in \Cref{algorithm:strong_composition}.
We start the description of \name's pseudocode by defining what constitutes certificates in \name.

\paragraph{Certificates.}
\name uses three types of certificates:
\begin{compactenum}
    \item \emph{Bucket certificates:} As the name suggests, these are the certificates obtained by \cert.
    (For the full details on bucket certificates, see \Cref{section:bucket_solution}.)

    \item \emph{Specific certificates:} These certificates vouch that one (and only one) value is safe.
    Concretely, a certificate $\Sigma$ is said to be \emph{specific} for a value $v$ if and only if $\Sigma$ is a $(t + 1)$-combined threshold signature for $v$ ($\mathsf{CombinedVerify}^{t + 1}(v, \Sigma) = \mathit{true}$).

    \item \emph{General certificates:} These certificates vouch that any value is safe: a general certificate proves that not all correct processes have proposed the same value.
    Formally, a certificate $\Sigma$ is said to be \emph{general} if and only if $\Sigma$ is a $(t + 1)$-combined threshold signature for ``any value'' ($\mathsf{CombinedVerify}^{t + 1}(\text{``any value''}, \Sigma) = \mathit{true}$).
\end{compactenum}
Crucially, any certificate (irrespectively of its type) contains only $O(1)$ words.

\paragraph{Protocol description.}
We explain the protocol from the perspective of a correct process $P_i$.
An execution of \name can be divided into two phases:
\begin{compactenum}
    \item Certification phase:
    In this phase, $P_i$ executes \cert.
    Namely, $P_i$ starts \cert with its proposal (line~\ref{line:start_bucket}), and then waits for \cert to stop (line~\ref{line:stop_bucket}).
    If $P_i$ acquires a bucket certificate $\Sigma$ for some value $v$, $P_i$ stores $v$ and $\Sigma$ in its local variables (line~\ref{line:fast_certificate_acquired_composition}), which concludes the certification phase of \name.

    \item Agreement phase:
    $P_i$ starts the agreement phase by broadcasting a \textsc{help\_req} message (line~\ref{line:broadcast_help_req_composition}) only if it has not previously acquired a certificate (due to the check at line~\ref{line_new:check_for_broadcast_help_req}).
    If $P_i$ receives a \textsc{help\_req} message in round 2 of the agreement phase (line~\ref{line:received_help_req}), $P_i$ replies via a \textsc{help\_reply} message (line~\ref{line:send_help_reply_composition}) which includes (1) an acquired certificate (if any), and (2) a partial signature of $P_i$'s proposal to allow for a (potential) creation of a specific certificate.

    $P_i$ executes the logic of rounds 3 and 4 only if it has not previously acquired a certificate (due to the checks at lines~\ref{line:check_3} and \ref{line:check_4}).
    In round 3, $P_i$ aims to acquire a certificate:
    \begin{compactitem}
        \item If $P_i$ receives a valid certificate in a \textsc{help\_reply} message (line~\ref{line:received_help_reply_check}), $P_i$ acquires the certificate (line~\ref{line:acquire_certificate_help_reply}).

        \item Otherwise, $P_i$ checks if it has received the same value from $t + 1$ processes via \textsc{help\_reply} messages (line~\ref{line:check_for_broadcast_final_certificate_composition}).
        If so, $P_i$ constructs a specific certificate by combining the received $t + 1$ partial signatures into a $(t + 1)$-combined threshold signature (line~\ref{line:obtain_specific_certificate_composition}).

        \item If neither of previous cases occurs, $P_i$ aims to build a general certificate.
        To this end, $P_i$ broadcasts a \textsc{allow-any} message (line~\ref{line:broadcast_allow_any_composition}) which carries a partial signature of the ``allow any'' string.
    \end{compactitem}
    In round 4, $P_i$ either receives a formed certificate from another process (line~\ref{line:receive_final_certificate}) or constructs a general certificate (line~\ref{line:construct_general_certificate}).
    At the end of round 4, $P_i$ proposes the acquired value-certificate pair to EBA (line~\ref{line:propose_to_adaptive_composition}).
    Finally, when $P_i$ decides from EBA (line~\ref{line:decide_eba}), it decides from \name (line~\ref{line:decide_strong_composition}).
\end{compactenum}

\begin{algorithm} [ht]
\caption{\name: Pseudocode (for process $P_i$)}
\label{algorithm:strong_composition}
\footnotesize
\begin{algorithmic} [1]

\State \textbf{Uses:}
\State \hskip2em EBA protocol with $O(n \cdot f)$ communication complexity (see~\cite{cohen2022make}), \textbf{instance} $\mathit{eba}$
\State \hskip2em \cert, a certification protocol with $O(n \cdot f)$ exchanged words (see \Cref{section:bucket_solution}), \textbf{instance} $\mathit{bucket}$

\smallskip
\State \textbf{Input Parameters:}
\State \hskip2em $\mathsf{Value}$ $\mathit{proposal}_i \gets P_i$'s proposal

\smallskip
\State \textbf{Variables:}
\State \hskip2em $\mathsf{Value}$ $\mathit{proposal\_to\_forward}_i \gets \bot$ \BlueComment{proposal to be forwarded to EBA}
\State \hskip2em $\mathsf{Boolean}$ $\mathit{certificate\_acquired}_i \gets \mathit{false}$ \BlueComment{states whether a certificate is acquired}
\State \hskip2em $\mathsf{Certificate}$ $\mathit{certificate}_i \gets \bot$ \BlueComment{acquired certificate}

\smallskip
\State \textcolor{blue}{\(\triangleright\) Certification phase}
\State \textbf{upon} $\mathsf{init}$: \BlueComment{start of the protocol}
\State \hskip2em \textbf{invoke} $\mathit{bucket}.\mathsf{start}(\mathit{proposal}_i)$ \label{line:start_bucket}
\State \hskip2em \textbf{wait for} $\mathit{bucket}.\mathsf{stop}$ \label{line:stop_bucket}

\State \hskip2em \textbf{if} a bucket certificate $\Sigma'$ for some value $v'$ is acquired:
\State \hskip4em $\mathit{proposal\_to\_forward}_i \gets v'$; $\mathit{certificate}_i \gets \Sigma'$; $\mathit{certificate\_acquired}_i \gets \mathit{true}$ \label{line:fast_certificate_acquired_composition}

\smallskip
\State \textcolor{blue}{\(\triangleright\) Agreement phase}
\State \textbf{Round 1:}
\State \hskip2em \textbf{if} $\mathit{certificate\_acquired}_i = \mathit{false}$: \label{line_new:check_for_broadcast_help_req}
\State \hskip4em \textbf{broadcast} $\langle \textsc{help\_req} \rangle$ \label{line:broadcast_help_req_composition}

\smallskip
\State \textbf{Round 2:}
\State \hskip2em \textbf{if} received $\langle \textsc{help\_req} \rangle$ from a process $P_j$: \label{line:received_help_req}
\State \hskip4em \textbf{send} $\langle \textsc{help\_reply}, \mathit{proposal\_to\_forward}_i, \mathit{certificate}_i, \mathsf{ShareSign}_i^{t + 1}(\mathit{proposal}_i) \rangle$ to $P_j$ \label{line:send_help_reply_composition}

\smallskip
\State \textbf{Round 3:} \label{line:round_3}
\State \hskip2em \textbf{if} $\mathit{certificate\_acquired}_i = \mathit{false}$: \label{line:check_3}

\State \hskip4em \textbf{if} received $\langle \textsc{help\_reply}, \mathsf{Value} \text{ } v, \mathsf{Certificate} \text{ } \Sigma, \cdot \rangle$ such that $\mathsf{validate}(v, \Sigma) = \mathit{true}$: \label{line:received_help_reply_check}
\State \hskip6em $\mathit{certificate\_acquired}_i \gets \mathit{true}$; $\mathit{proposal\_to\_forward}_i \gets v$; $\mathit{certificate}_i \gets \Sigma$ \label{line:acquire_certificate_help_reply}

\State \hskip4em \textbf{else if} exists $\mathsf{Value}$ $v$ with $t + 1$ $\langle \textsc{help\_reply}, v, \cdot, \cdot \rangle$ messages received: \label{line:check_for_broadcast_final_certificate_composition} 
\State \hskip6em $\mathit{certificate\_acquired}_i \gets \mathit{true}$; $\mathit{proposal\_to\_forward}_i \gets v$ 
\State \hskip6em $\mathit{certificate}_i \gets \mathsf{Combine}^{t + 1}\big( \{ \mathit{sig} \,|\, \mathit{sig} \text{ is received in the aforementioned $t + 1$ \textsc{help\_reply} messages}   \} \big)$ \label{line:obtain_specific_certificate_composition}
\State \hskip6em \textbf{broadcast} $\langle \textsc{final\_certificate}, \mathit{proposal\_to\_forward}_i, \mathit{certificate}_i \rangle$

\State \hskip4em \textbf{else} \label{line:check_for_broadcast_allow_any_composition}
\State \hskip6em \textbf{broadcast} $\langle \textsc{allow-any}, \mathsf{ShareSign}_i^{t + 1}(\text{``allow any''}) \rangle$ \label{line:broadcast_allow_any_composition}

\smallskip
\State \textbf{Round 4:} \label{line:round_4}
\State \hskip2em \textbf{if} $\mathit{certificate\_acquired}_i = \mathit{false}$: \label{line:check_4}
\State \hskip4em \textbf{if} received $\langle \textsc{final\_certificate}, \mathsf{Value} \text{ } v, \mathsf{Certificate} \text{ } \Sigma \rangle$ such that $\mathsf{validate}(v, \Sigma) = \mathit{true}$: \label{line:receive_final_certificate}
\State \hskip6em $\mathit{proposal\_to\_forward}_i \gets v$; $\mathit{certificate}_i \gets \Sigma$; $\mathit{certificate\_acquired}_i \gets \mathit{true}$
\State \hskip4em \textbf{else:}
\State \hskip6em $\mathit{proposal\_to\_forward}_i \gets \mathit{proposal}_i$
\State \hskip6em $\mathit{certificate}_i \gets \mathsf{Combine}^{t + 1}\big( \{ \mathit{sig} \,|\, \mathit{sig} \text{ is received in the \textsc{allow-any} messages}   \} \big)$ \label{line:construct_general_certificate}\label{line:obtain_general_certificate_composition}
\State \hskip2em \textbf{invoke} $\mathit{eba}.\mathsf{propose}(\mathit{proposal\_to\_forward}_i, \mathit{certificate}_i)$ \label{line:propose_to_adaptive_composition}

\smallskip
\State \textbf{upon} $\mathit{eba}.\mathsf{decide}(\mathsf{Value} \text{ } v, \mathsf{Certificate} \text{ } \Sigma)$: \label{line:decide_eba}
\State \hskip2em \textbf{trigger} $\mathsf{decide}(v)$ \label{line:decide_strong_composition}
\end{algorithmic}
\end{algorithm}

\subsection{Proof of Correctness \& Complexity} \label{subsection:strong_proof}




\paragraph{Proof of correctness.}
First, we prove that \name satisfies strong validity.
Therefore, we start by showing that, if all correct processes propose the same value, all correct processes obtain a certificate by the beginning of round 3.

\begin{lemma} \label{lemma:by_round_3}
If all correct processes propose the same value, then every correct process obtains a certificate by the beginning of round 3. 
\end{lemma}
\begin{proof}
Let all correct processes propose the same value $v$.
We distinguish two possible cases:
\begin{compactitem}
    \item Let there exist a correct process $P_i$ which acquires a bucket certificate before round 1 (at line~\ref{line:fast_certificate_acquired_composition}).
    In this case, every correct process obtains a certificate by the beginning of round 3 as $P_i$ shares the acquired certificate in a \textsc{help\_reply} message (line~\ref{line:send_help_reply_composition}).
    Thus, the statement of the lemma holds in this scenario.
    
    \item Let no correct process obtain a fast certificate before round 1.
    In this case, every correct process broadcasts a \textsc{help\_req} message (line~\ref{line:broadcast_help_req_composition}).
    Thus, every correct process receives \textsc{help\_reply} messages from all other correct processes, which implies that every correct process receives (at least) $t + 1$ partial signatures of $v$ at the beginning of round 3 (line~\ref{line:check_for_broadcast_final_certificate_composition}).
    Hence, every correct process obtains a specific certificate by the beginning of round 3 (line~\ref{line:obtain_specific_certificate_composition}).
\end{compactitem}
As the statement of the lemma holds in both possible cases, the proof is concluded.
\end{proof}

We are now ready to prove that \name satisfies strong validity.

\begin{lemma} \label{lemma:strong_validity}
\name (\Cref{algorithm:strong_composition}) satisfies strong validity.
\end{lemma}
\begin{proof}
By contradiction, suppose that all correct processes propose the same value $v$, and a correct process decides some value $v' \neq v$ (line~\ref{line:decide_strong_composition}).
Thus, there exists a certificate $\Sigma'$ for $v'$.
Let us consider all possible cases:
\begin{compactitem}
    
    \item Let $\Sigma'$ be a specific certificate.
    In this case, $\mathsf{CombinedVerify}^{t + 1}(v', \Sigma') = \mathit{true}$.
    Hence, a correct process has broadcast a \textsc{help\_reply} message containing a (partial) signature of $v'$ (line~\ref{line:send_help_reply_composition}).
    However, this is impossible as all correct processes propose the same value $v \neq v'$.
    Thus, this case is impossible, as well.
    
    \item Let $\Sigma'$ be a general certificate.
    Therefore, $\mathsf{CombinedVerify}^{t + 1}(\text{``any value''}, \Sigma') = \mathit{true}$.
    Hence, a correct process $P_i$ has broadcast an \textsc{allow-any} message (line~\ref{line:broadcast_allow_any_composition}).
    Thus, $P_i$ has not obtained a certificate by the beginning of round 3 (due to the check at line~\ref{line:check_3}).
    However, this contradicts \Cref{lemma:by_round_3}, implying that this case cannot occur.

    \item Let $\Sigma'$ be a bucket certificate.
    However, the existence of $\Sigma'$ is impossible due to the safety property of \cert.
    Thus, this case is impossible.
\end{compactitem}
As no case is possible, there does not exist a certificate $\Sigma'$ for $v'$.
Thus, we reach a contradiction, and the lemma holds.
\end{proof}

Finally, we are ready to prove the correctness of \name.

\begin{theorem}
\name (\Cref{algorithm:strong_composition}) is correct.
\end{theorem}
\begin{proof}
Termination of \name follows from (1) the fact that every correct process eventually proposes to EBA (line~\ref{line:propose_to_adaptive_composition}), and (2) termination of EBA.
Similarly, agreement of \name follows directly from agreement of EBA.
Lastly, strong validity of \name follows from \Cref{lemma:strong_validity}.
\end{proof}

\paragraph{Proof of complexity.}
We start by showing that, if $f \leq t_o$, no correct process sends more than $O(f)$ words from round 1 onward (without counting the messages sent in EBA).

\begin{lemma} \label{lemma:f_messages}
If $f \leq t_o$, no correct process sends more than $O(f)$ words from round 1 onward excluding the words sent in EBA.
\end{lemma}
\begin{proof}
Let us fix any correct process $P$.
If $f \leq t_o$, $P$ acquires a bucket certificate (line~\ref{line:fast_certificate_acquired_composition}) before starting round 1 (by the liveness property of \cert).
Hence, after the attainment of the bucket certificate, $P$ can only send $\textsc{help\_reply}$ messages in round 2 (line~\ref{line:send_help_reply_composition}).
As no correct process ever broadcasts a \textsc{help\_req} message (line~\ref{line:broadcast_help_req_composition}), $P$ sends (at most) $f$ \textsc{help\_reply} messages (line~\ref{line:send_help_reply_composition}).
Given that every \textsc{help\_reply} message has $O(1)$ words, the lemma holds.
\end{proof}

Next, we show that no correct process sends more than $O(n)$ words from round 1 onward (without counting the messages sent in EBA).

\begin{lemma} \label{lemma:n_messages}
No correct process sends more than $O(n)$ words from round 1 onward excluding the words sent in EBA.
\end{lemma}
\begin{proof}
Follows directly from \Cref{algorithm:strong_composition}.
\end{proof}

Finally, we are ready to prove the word complexity of \name.

\begin{theorem}
\name (\Cref{algorithm:strong_composition}) achieves $O(n \cdot f)$ word complexity.
\end{theorem}
\begin{proof}
We consider two cases:
\begin{compactitem}
    \item Let $f \leq t_{o}$.
    In this case, correct processes send $O(n \cdot f)$ words during the certification phase.
    From round 1 onward, each correct process sends at most $O(f)$ words (by \Cref{lemma:f_messages}).
    Finally, as the word complexity of EBA is $O(n \cdot f)$, the overall word complexity is $O(n \cdot f)$ in this case.
    
    \item Otherwise, $f > t_o$, which implies that $f \in O(n)$.
    During the certification phase, all correct processes send (at most) $O(n^2)$ words.
    Moreover, from round 1 onward, each correct process sends at most $O(n)$ words (by \Cref{lemma:n_messages}).
    Lastly, in EBA, correct processes exchange (at most) $O(n^2)$ words.
    Thus, in this case, the word complexity of \name is $O(n^2)$, which is indeed $O(n \cdot f)$ as $f \in O(n)$.
\end{compactitem}
As the word complexity of \name is $O(n \cdot f)$ in both cases, the theorem holds.
\end{proof}
\section{\cert: Certification Protocol} \label{section:bucket_solution}

In this section, we describe \cert, a building block of \name (\Cref{algorithm:strong_composition}).
\cert is a certification protocol which exchanges $O(n \cdot f)$ words.

\subsection{Pseudocode}

The pseudocode of \cert is given in \Cref{algorithm:certificate_creation,algorithm:bucket_plus}.
First, we describe bucket certificates.

\paragraph{Bucket certificates.}
\cert produces two types of (bucket) certificates:
\begin{compactenum}
    \item \emph{Positive certificates:}
    Positive certificates, same as specific certificates in \name (see \Cref{subsection:strong_pseudocode}), vouch that one specific value is safe.
    Namely, a certificate $\Sigma^+$ is said to be a \emph{positive bucket certificate} for a value $v$ if and only if $\Sigma^+$ is a $(t + 1)$-combined threshold signature for $v$.

    \item \emph{Negative certificates:}
    Negative certificates, same as general certificates in \name (see \Cref{subsection:strong_pseudocode}), prove that not all correct processes have proposed the same value.
    However, negative bucket certificates have a more complex format than the general certificates in \name.

    In order to properly define negative certificate, we introduce \emph{groups}.
    A group $g$ is a tuple $(x, y, \mathit{tsignature})$, where $x, y \in \mathsf{Value}$ and $\mathit{tsignature}$ is a threshold signature, such that (1) $x < y$, and (2) $\mathit{tsignature}$ is a $(t + 1)$-combined threshold signature of $(x, y)$.
    Lastly, a negative certificate is a tuple $\big(g_1 = (x_1, y_1, \mathit{tsignature}_1), g_2 = (x_2, y_2, \mathit{tsignature}_2), ..., g_k = (x_k, y_k, \mathit{tsignature}_k)  \big)$ of $k > 0$ groups such that:
    \begin{compactitem}
        \item for every $i \in [2, k]$, $y_{i - 1} = x_i$, and

        \item $x_1 = v_{\mathit{min}}$, and

        \item $y_k = v_{\mathit{max}}$.
    \end{compactitem}
\end{compactenum}

\begin{algorithm} [ht]
\caption{$\mathsf{certificate\_creation}(\mathit{start\_value}_i, L)$: Pseudocode (for process $P_i$)}
\label{algorithm:certificate_creation}
\footnotesize
\begin{algorithmic} [1]

\State \textbf{Input Parameters:}
\State \hskip2em $\mathsf{Value}$ $\mathit{start\_value}_i$ \BlueComment{the input value; first parameter}
\State \hskip2em $\mathsf{Process}$ $L$ \BlueComment{the leader; second parameter}

\smallskip
\State \textbf{Variables:}
\State \hskip2em $\mathsf{Boolean}$ $\mathit{certificate\_acquired}_i \gets \mathit{false}$ \BlueComment{$\mathit{true}$ if a certificate is acquired; otherwise, $\mathit{false}$}
\State \hskip2em $\mathsf{Set}(\mathsf{Group})$ $\mathit{groups}_i \gets \bot$ \BlueComment{partitioning groups}

\smallskip
\State \textbf{Round 1:}
\State \hskip2em \textcolor{blue}{\(\triangleright\) inform the leader of the starting value}
\State \hskip2em \textbf{send} $\langle \textsc{disclose}, \mathit{start\_value}_i, \mathsf{ShareSign}^{t+1}_i(\mathit{start\_value}_i) \rangle$ to $L$ \label{line:send_disclose_leader}

\smallskip
\State \textbf{Round 2:} \BlueComment{executed only by $L$}
\State \hskip2em \textbf{if} $P_i = L$:
\State \hskip4em Let $\mathit{discloses}$ denote the set of received \textsc{disclose} messages
\State \hskip4em \textbf{if} exists $\mathsf{Value} \text{ } v$ such that $\langle \textsc{disclose}, v, \mathsf{P\_Signature} \text{ } \mathit{sig}\rangle$ is received from $t+1$ processes: \label{line:t+1_disclose_check}
\State \hskip6em \textcolor{blue}{\(\triangleright\) a positive certificate for $v$ can be constructed}
\State \hskip6em Let $\mathit{positive\_certificate} \gets \mathsf{Combine}^{t + 1}(\{\mathit{sig} \,|\, \mathit{sig} \text{ is received in the } t + 1 \text{ received } \textsc{disclose} \text{ messages}\})$ \label{line:combine_certification_phase}
\State \hskip6em $\mathit{certificate\_acquired}_i \gets \mathit{true}$
\State \hskip6em \textbf{broadcast} $\langle \textsc{certificate}, v, \mathit{positive\_certificate} \rangle$  \label{line:broadcast_positive_certificate}

\State \hskip4em \textcolor{blue}{\(\triangleright\) check if enough messages are received}
\State \hskip4em \textbf{else if} $|\mathit{discloses}| \geq n - t_o$: \label{line:else_2t+1_disclose}
\State \hskip6em \textcolor{blue}{\(\triangleright\) a positive certificate cannot be constructed; start partitioning}
\State \hskip6em $\mathit{groups}_i \gets \mathsf{partition}(\mathit{discloses})$ \label{line:partitioning}
\State \hskip6em \textbf{broadcast} $\langle \textsc{partition\_req},  \mathit{groups_i}\rangle$ \label{line:broadcast_groups} 

\smallskip
\State \textbf{Round 3:}
\State \hskip2em \textbf{if} $\langle \textsc{partition\_req}, \mathsf{Set(Group)} \text{ } \mathit{groups} \rangle$ is received from $L$ such that
$|\mathit{groups}| \in O(1)$:
\State \hskip4em Let $\mathit{negatives} \gets \emptyset$
\State \hskip4em \textbf{for each} $\big(g = (x, y, \cdot)\big) \in \mathit{groups}$ such that $\neg(x \leq \mathit{start\_value}_i < y)$:
\State \hskip6em $\mathit{negatives} \gets \mathit{negatives} \cup (g, \mathsf{ShareSign}^{t+1}_i(x, y))$ \label{line:partial_signature_partition}
\State \hskip4em \textbf{send} $\langle \textsc{partition\_reply}, \mathit{negatives} \rangle $ to $L$ \label{line:share_sign_description}



\smallskip
\State \textbf{Round 4:}
\State \hskip2em \textbf{if} $P_i = L$ and $\mathit{certificate\_acquired}_i = \mathit{false}$: \BlueComment{executed only by $L$}
\State \hskip4em Let $\mathit{replies}$ denote the set of received \textsc{partition\_reply} messages \label{line:receive_partition_reply}
\State \hskip4em Let $\mathsf{Certificate}$ $\mathit{negative\_certificate} \gets \mathsf{construct\_negative\_certificate}(\mathit{replies}, \mathit{groups}_i)$  \label{line:compute_negative_certificate}
\State \hskip4em \textbf{if} $\mathit{negative\_certificate} \neq \bot$: \BlueComment{negative certificate is successfully created} \label{line:check_success_negative_certificate}
\State \hskip6em $\mathit{certificate\_acquired}_i \gets \mathit{true}$
\State \hskip6em \textbf{broadcast} $\langle \textsc{certificate}, \mathsf{default\_value}, \mathit{negative\_certificate} \rangle$ \BlueComment{$\mathsf{default\_value} \in \mathsf{Value}$ is a constant} \label{line:broadcast_negative_certificate}

\smallskip
\State \textbf{Round 5:}
\State \hskip2em \textbf{if} $\langle \textsc{certificate}, \mathsf{Value} \text{ } \mathit{value}, \mathsf{Certificate} \text{ } \Sigma \rangle$ is received from $L$ such that $\mathsf{validate}(\mathit{value}, \Sigma) = \mathit{true}$: \label{line:received_certificate}
\State \hskip4em \textbf{return} $\mathit{value}, \Sigma$ \label{line:return_1}
\State \hskip2em \textbf{else:}
\State \hskip4em \textbf{return} $\bot, \bot$ \label{line:return_2}

\end{algorithmic}
\end{algorithm}

\begin{algorithm} [ht]
\caption{\cert: Pseudocode (for process $P_i$)}
\label{algorithm:bucket_plus}
\footnotesize
\begin{algorithmic} [1]
\State \textbf{Input Parameters:}
\State \hskip2em $\mathsf{Value}$ $\mathit{start\_value}_i \gets v$, where $P_i$ has previously invoked $\mathsf{start}(v)$ 

\smallskip
\State \textbf{Variables:}
\State \hskip2em $\mathsf{Value}$ $\mathit{certified\_value}_i \gets \bot$
\State \hskip2em $\mathsf{Certificate}$ $\mathit{certificate}_i \gets \bot$
\State \hskip2em $\mathsf{Boolean}$ $\mathit{certificate\_acquired}_i \gets \mathit{false}$
\State \hskip2em $\mathsf{Boolean}$ $\mathit{started\_creation}_i = \mathit{false}$

\medskip
\State \textbf{for} $j = 1$ to $t_o + 1$:

\smallskip
\State \hskip2em \textbf{Round 1:}
\State \hskip4em \textbf{if} $\mathit{certificate\_acquired}_i= \mathit{false}$: \label{line:check_certificate_acquired_1}
\State \hskip6em \textbf{if} $P_i = P_j$: \BlueComment{if $P_j$ is the leader of this iteration}
\State \hskip8em \textbf{broadcast} $\langle \textsc{aid\_req} \rangle$ \label{line:bcast_aid_req}
\State \hskip6em \textbf{else}:
\State \hskip8em \textbf{send} $\langle \textsc{aid\_req} \rangle$ to $P_j$ \label{line:send_aid_req}

\smallskip
\State \hskip2em \textbf{Round 2:}
\State \hskip4em \textbf{if} $\langle \textsc{aid\_req} \rangle$ is received from $P_j$: \label{line:received_aid_req}
\State \hskip6em $\mathit{started\_creation}_i = \mathit{true}$
\State \hskip6em \textbf{invoke} $\mathsf{certificate\_creation}(\mathit{start\_value}_i, P_j)$ \label{line:start_Bucket_instance}
\State \hskip4em \textbf{else if} $P_i = P_j$ and $\langle \textsc{aid\_req} \rangle$ is received from some process $P_k$:
\State \hskip6em \textbf{send} $\langle \textsc{aid\_reply}, \mathit{certified\_value}_i, \mathit{certificate}_i \rangle$ \text{to} $P_k$ \label{line:send_aid_reply}

\smallskip
\State \hskip2em \textbf{Round 3:}
\State \hskip4em \textbf{if} $\mathit{started\_creation}_i = \mathit{false}$ and $\langle \textsc{aid\_reply}, \mathsf{Value} \text{ } v', \mathsf{Certificate} \text{ } \Sigma' \rangle$ is received from $P_j$ such that $\mathsf{validate}(v', \Sigma') = \mathit{true}$: \label{line:received_aid_reply}
\State \hskip6em
$\mathit{certified\_value}_i \gets v'$; $\mathit{certificate}_i \gets \Sigma'$ \label{line:certificate_acquired_aid_reply}

\smallskip
\State \hskip2em \textbf{Round 6:}
\State \hskip4em $\mathit{certified\_value}_i, \mathit{certificate}_i \gets $ outputs of $\mathsf{certificate\_creation}$ (if previously invoked) \label{line:certificate_creation_return}
\State \hskip4em \textbf{if} $\mathit{certificate}_i \neq \bot$:
\State \hskip6em $\mathit{certificate\_acquired}_i \gets \mathit{true}$
\State \hskip6em \textbf{trigger} $\mathsf{acquire}(\mathit{certified\_value}_i, \mathit{certificate}_i)$

\smallskip
\State \textbf{trigger} $\mathsf{stop}$ \label{line:bucket_stop}


\end{algorithmic}
\end{algorithm}

\paragraph{$\mathsf{certificate\_creation}$ (\Cref{algorithm:certificate_creation}).}
The input parameters of the $\mathsf{certificate\_creation}$ subprotocol are $\mathit{start\_value}_i$, the value with which $P_i$ invoked $\cert.\mathsf{start}(\mathsf{start\_value}_i)$, and $L$, the leader of the specific $\mathsf{certificate\_creation}$ invocation.
When a correct process $P_i$ invokes $\mathsf{certificate\_creation}$, it first informs the leader $L$ of its $\mathit{start\_value}_i$ by sending a \textsc{disclose} message to $L$ (line~\ref{line:send_disclose_leader}); the \textsc{disclose} message contains $\mathit{start\_value}_i$ and a partial signature of $\mathit{start\_value}_i$.
When $L$ receives the previously sent \textsc{disclose} messages, it first checks (if $L$ is correct) if it has received the same value from (at least) $t + 1$ processes (line~\ref{line:t+1_disclose_check}).
If so, $L$ constructs a positive bucket certificate by combining the received partial signatures into a $(t + 1)$-combined threshold signature (line~\ref{line:combine_certification_phase}).
Otherwise, $L$ checks if it has received $n - t_o \geq 2t + 1$ \textsc{disclose} messages (line~\ref{line:else_2t+1_disclose}).
If it has, $L$ starts \emph{partitioning} the received values into groups (line~\ref{line:partitioning}) such that:
\begin{compactitem}
    \item For every constructed group $g = (x, y, \cdot)$, the values of at most $t$ processes ``fall into'' $g$, i.e., there exist at most $t$ processes $P_j$ such that $P_j$'s value $\mathit{value}_j$ is such that $x \leq \mathit{value}_j < y$.

    \item For any two adjacent constructed groups $g = (x, y, \cdot)$ and $g' = (x', y', \cdot)$, the values of at least $t + 1$ processes ``fall into'' $g$ or $g'$, i.e., there exist at least $t + 1$ processes $P_j$ such that $P_j$'s value $\mathit{value}_j$ is such that $x \leq \mathit{value}_j < y$ or $x' \leq \mathit{value}_j < y'$.
\end{compactitem}
Importantly, $L$ is able to partition all received values into $O(1)$ groups, which ensures that a negative certificate contains $O(1)$ words (as $O(1)$ groups will eventually be accompanied by a threshold signature).
Once $L$ has partitioned all the values into groups, it sends the groups to all processes (line~\ref{line:broadcast_groups}) in order to form accompanying threshold signatures.
When $P_i$ receives the formed groups, $P_i$ sends back a partial signature of every group to which its value \emph{does not} belong (lines~\ref{line:partial_signature_partition} and~\ref{line:share_sign_description}).

Once $L$ receives partial signatures from the other processes (line~\ref{line:receive_partition_reply}), $L$ tries to create a negative certificate (line~\ref{line:compute_negative_certificate}).
(If $L$ has previously constructed a positive certificate, this part of the $\mathsf{certification\_creation}$ subprotocol is not executed.)
If $L$ successfully creates a negative certificate (line~\ref{line:check_success_negative_certificate}, it disseminates it to every process (line~\ref{line:broadcast_negative_certificate}).
Finally, if $P_i$ receives a valid (positive or negative) certificate from $L$ (line~\ref{line:received_certificate}), $P_i$ returns that certificate (line~\ref{line:return_1}).
Otherwise, $P_i$ returns $\bot$ (line~\ref{line:return_2}).

\paragraph{\cert's Protocol (\Cref{algorithm:bucket_plus}).}
\cert operates in $t_o + 1$ iteration such that each iteration takes up to 6 rounds and has its unique leader.
We explain \cert's pseudocode from the perspective of a correct process $P_i$.

In the first round of each iteration, $P_i$ checks if it has already acquired a certificate (line~\ref{line:check_certificate_acquired_1}).
If it has not, $P_i$ broadcasts an \textsc{aid\_req} message (line~\ref{line:bcast_aid_req}) if $P_i$ is the leader of the iteration.
Otherwise, $P_i$ sends an \textsc{aid\_req} message to the leader (line~\ref{line:send_aid_req}). 

If $P_i$ receives an \textsc{aid\_req} message from the leader (line~\ref{line:received_aid_req}), $P_i$ initiates the $\mathsf{certificate\_creation}$ subprotocol (line~\ref{line:start_Bucket_instance}).
Otherwise, if $P_i$ is the leader and it has received an \textsc{aid\_req} from any process $P_k$, then $P_i$ sends the certificate it has already acquired $P_k$ (line~\ref{line:send_aid_reply}).
(Note that $P_i$ must have acquired a certificate as it has not previously broadcast and received an \textsc{aid\_req} message.)

In the third round, if $P_i$ has not invoked $\mathsf{certificate\_creation}$ and it has received a valid certificate from the leader (line~\ref{line:received_aid_reply}), $P_i$ acquires a certificate (line~\ref{line:certificate_acquired_aid_reply}).
In round 6, if $P_i$ has previously invoked $\mathsf{certificate\_creation}$, $P_i$ stores the certificate returned by $\mathsf{certificate\_creation}$ (line~\ref{line:certificate_creation_return}).
Finally, once all $t_o + 1$ iterations terminate, $P_i$ stops \cert (line~\ref{line:bucket_stop}).

\subsection{Proof of Correctness \& Complexity}




\paragraph{Proof of correctness.}
We start by proving that \cert satisfies safety.

\begin{lemma} \label{lemma:certification_safety}
\cert (\cref{algorithm:certificate_creation,algorithm:bucket_plus}) satisfies safety.
\end{lemma}
\begin{proof}
By contradiction, suppose that all correct processes have started certification with the same value $v$, and that a (correct or faulty) process obtains a certificate $\Sigma'$ for some value $v' \neq v$.
Let us consider two possible scenarios:
\begin{compactitem}
    \item Let $\Sigma'$ is a positive certificate.
    Hence, $\Sigma'$ is a $(t + 1)$-combined threshold signature of $v'$, which implies that (at least) one correct process have sent a \textsc{disclose} message for $v'$ (line~\ref{line:send_disclose_leader} of \Cref{algorithm:certificate_creation}).
    However, this is impossible as no correct process has started certification with a value different from $v$ (line~\ref{line:start_Bucket_instance} of \Cref{algorithm:bucket_plus}).
    Thus, this scenario is impossible.

    \item Let $\Sigma'$ be a negative certificate.
    Let $g = (x, y, \cdot)$ denote the group incorporated into $\Sigma'$ such that $x \leq v < y$.
    As $g$ includes a $(t + 1)$-combined threshold signature of $(x, y)$, (at least) one correct process participated in creating the signature (line~\ref{line:share_sign_description} of \Cref{algorithm:certificate_creation}).
    However, this cannot happen as all correct processes have started certification with $v$.
    Therefore, this scenario is impossible as well.
\end{compactitem}
As neither of the two scenarios is possible, we reach a contradiction and the lemma holds.
\end{proof}

Next, we prove \cert's liveness.
To this end, we show that, if all correct processes participate in an iteration $\mathcal{I}$ of the $\mathsf{certificate\_creation}$ subprotocol, the leader of $\mathcal{I}$ is correct and $f \leq t_o$, all correct processes obtain a certificate in $\mathcal{I}$.

\begin{lemma} \label{lemma:iteration_certificate}
Consider any specific iteration $\mathcal{I}$ of the $\mathsf{certificate\_creation}$ subprotocol (\Cref{algorithm:certificate_creation}) such that (1) all correct processes participate in $\mathcal{I}$, (2) the leader $L$ of $\mathcal{I}$ is correct, and (3) $f \leq t_o$.
All correct processes obtain a certificate in $\mathcal{I}$.
\end{lemma}
\begin{proof}
All correct processes send their values to $L$ (line~\ref{line:send_disclose_leader}).
There exist two possibilities:
\begin{compactitem}
    \item $L$ receives the identical value from (at least) $t + 1$ processes; let that value be $v$.
    In this case, $L$ is able to create a positive certificate (line~\ref{line:combine_certification_phase}), which is then disseminated (line~\ref{line:broadcast_positive_certificate}).
    Thus, the statement of the lemma holds in this case.

    \item $L$ does not receive the identical value from $t + 1$ processes.
    As all correct processes participate in $\mathcal{I}$ and $f \leq t_o$, the number of received values is $\geq n - t_o$, which implies that the check at line~\ref{line:else_2t+1_disclose} passes.
    Hence, $L$ indeed partitions all the received values into groups (line~\ref{line:partitioning}), and disseminates the groups to all processes via \textsc{partition\_req} messages (line~\ref{line:broadcast_groups}).

    As $n - t_o \geq 2t + 1$ and each formed groups includes value of at most $t$ processes, each group \emph{does not} contain values from (at least) $t + 1$ correct processes.
    Hence, for each formed group, at least $t + 1$ correct processes partially sign the group's description (line~\ref{line:partial_signature_partition}), and send it back to $L$ (line~\ref{line:share_sign_description}).
    Therefore, $L$ successfully constructs a negative certificate (line~\ref{line:compute_negative_certificate}), which $L$ then disseminates to all processes (line~\ref{line:broadcast_negative_certificate}).
    THe statement of the lemma holds in this case as well.
\end{compactitem}
As the statement of the lemma holds in both possible cases, the proof is concluded.
\end{proof}



The next lemma proves that, if $f \leq t_o$, all correct processes obtain a certificate in the first iteration of the for loop in \Cref{algorithm:bucket_plus} with a correct leader.

\begin{lemma} \label{lemma:liveness_helper}
Let $f \leq t_o$.
All correct processes obtain a certificate in the first iteration $\mathcal{J}$ of the for loop in \Cref{algorithm:bucket_plus} such that the leader $P_{\mathcal{J}}$ of $\mathcal{J}$ is correct.
\end{lemma}
\begin{proof}
If $P_{\mathcal{J}}$ has not previously obtained a certificate, every correct process starts (in $\mathcal{J}$) the $\mathsf{certificate\_creation}$ subprotocol (line~\ref{line:start_Bucket_instance}), and every correct processes obtains a certificate (by \Cref{lemma:iteration_certificate}).
Otherwise, $P_{\mathcal{J}}$ ``helps'' every correct process which has not obtained a certificate by sending $P_{\mathcal{J}}$'s certificate via an \textsc{aid\_reply} message (line~\ref{line:send_aid_reply}).
Thus, the statement of the lemma is satisfied even in this case.  
\end{proof}

We are ready to prove \cert's liveness.

\begin{lemma} \label{lemma:certification_liveness}
\cert (\cref{algorithm:certificate_creation,algorithm:bucket_plus}) satisfies liveness.
\end{lemma}
\begin{proof}
Follows directly from \Cref{lemma:liveness_helper}.
\end{proof}

Finally, we are ready to prove that \cert is correct.

\begin{theorem}
\cert (\cref{algorithm:certificate_creation,algorithm:bucket_plus}) is correct.
\end{theorem}
\begin{proof}
Safety of \cert follows directly from \Cref{lemma:certification_safety}, whereas its liveness follows from \Cref{lemma:certification_liveness}.
Finally, termination follows the fact that both \cref{algorithm:certificate_creation,algorithm:bucket_plus} terminate within a finite number of rounds.
\end{proof}

\paragraph{Proof of complexity.}
First, we prove that a correct process is (always) able to partition $\leq n - t_o$ values into $O(1)$ groups such that (1) each groups contains at most $t$ values, and (2) any two adjacent groups together contain at least $t + 1$ values.

\begin{lemma} \label{lemma:groups}
Consider a correct process $L$ which has received $x$ values, with $n - t_o \leq x \leq n$.
$L$ is able to partition all the received values into $O(1)$ groups such that (1) each group contains at most $t$ received values, and (2) any two adjacent groups contain at least $t + 1$ received values.
\end{lemma}
\begin{proof}
To prove the lemma, we reuse the argument from~\cite{rambaud2022linear}.
As any two adjacent groups contain (collectively) at least $t + 1$ values, and there are at most $3t + 1$ values, there can be at most $5 \in O(1)$ groups.
Indeed, the first two groups contain (at least) $t + 1$ values.
The same holds for the groups $3$ and $4$.
Hence, there are $3t + 1 - 2(t + 1) < t$ values left for group $5$, which concludes the lemma.
\end{proof}

A direct consequence of \Cref{lemma:groups} is that any bucket certificate has $O(1)$ words.

\begin{lemma} \label{lemma:certificate_words}
Any (positive or negative) bucket certificate has $O(1)$ words.
\end{lemma}
\begin{proof}
A positive certificate has $O(1)$ words as it is a threshold signature, whereas a negative certificate has $O(1)$ words as it is a set of $O(1)$ (by \Cref{lemma:groups}) threshold signatures, one for each formed group.
\end{proof}

Next, we prove that only $O(n)$ words are exchanged by all correct processes in any iteration of the $\mathsf{certificate\_creation}$ subprotocol.

\begin{lemma} \label{lemma:words_in_iteration}
Consider any specific iteration $\mathcal{I}$ of the $\mathsf{certificate\_creation}$ subprotocol (\Cref{algorithm:certificate_creation}).
All correct processes collectively send $O(n)$ words in $\mathcal{I}$.
\end{lemma}
\begin{proof}
If the leader is correct, it sends $O(n)$ words in $\mathcal{I}$.
Note that each \textsc{partition\_req} message contains $O(1)$ words as $O(1)$ groups are formed (by \Cref{lemma:groups}).
Moreover, each \textsc{certificate} message contains $O(1)$ words as any bucket certificate has $O(1)$ words (by \Cref{lemma:certificate_words}).
Furthermore, each non-leader process sends only $O(1)$ words in $\mathcal{I}$.
Thus, all correct processes collectively send $O(n) + n \cdot O(1) = O(n)$ words in $\mathcal{I}$.
\end{proof}

Next, we show that all correct processes collectively send $O(n)$ words in each iteration of the for loop in \Cref{algorithm:bucket_plus}.

\begin{lemma} \label{lemma:words_in_bigger_iteration}
Consider any specific iteration $\mathcal{J}$ of the for loop in \Cref{algorithm:bucket_plus}.
All correct processes collectively send $O(n)$ words in $\mathcal{J}$.
\end{lemma}
\begin{proof}
If the leader of $\mathcal{J}$ (process $P_{\mathcal{J}}$) is correct, it sends $O(n)$ words as each \textsc{aid\_reply} message contains $O(1)$ words (by \Cref{lemma:certificate_words}).
Other processes send $O(1)$ words.
Finally, given that all correct processes send $O(n)$ words in the corresponding iteration of the $\mathsf{certificate\_creation}$ subprotocol (by \Cref{lemma:words_in_iteration}), the lemma holds.
\end{proof}

Next, we prove that $O(f)$ words are sent in any iteration of the for loop such that (1) its leader is correct, and (2) all correct processes have previously acquired a certificate.

\begin{lemma} \label{lemma:silence}
Consider any specific iteration $\mathcal{J}$ of the for loop in \Cref{algorithm:bucket_plus} such that (1) the leader ($P_{\mathcal{J}}$) is correct, and (2) all correct processes have acquired a certificate prior to $\mathcal{J}$.
All correct processes send $O(f)$ words in $\mathcal{J}$.
\end{lemma}
\begin{proof}
No correct process invokes the $\mathsf{certificate\_creation}$ subprotocol as $P_{\mathcal{J}}$ never sends an \textsc{aid\_req} message.
Hence, only $P_{\mathcal{J}}$ (out of the correct processes) sends messages in $\mathcal{J}$.
Namely, $P_{\mathcal{J}}$ replies to each received $\textsc{aid\_req}$ message.
As there can only be $O(f)$ such messages (each from a faulty process) and any certificate contains $O(1)$ words (by \Cref{lemma:certificate_words}), $P_{\mathcal{J}}$ sends $O(f)$ words, which proves the lemma.
\end{proof}

Lastly, we are ready to prove that \cert exchanges $O(n \cdot f)$ words.

\begin{theorem}
\cert (\cref{algorithm:certificate_creation,algorithm:bucket_plus}) exchanges $O(n \cdot f)$ words.
\end{theorem}
\begin{proof}
We separate all the iterations of the for loop into two disjoint classes: (1) class $\mathcal{J}^+$ with a correct leader, and (2) class $\mathcal{J}^-$ with a faulty leader.
Note that $|\mathcal{J}^+| \in O(n)$ (as $t_o \in O(n)$) and $|\mathcal{J}^-| \in O(f)$.
We separate the proof into two possible cases:
\begin{compactitem}
    \item Let $f \leq t_o$.
    All correct processes exchange $O(n)$ words in each iteration which belongs to $\mathcal{J}^-$ (by \Cref{lemma:words_in_bigger_iteration}).
    In the first iteration which belongs to $\mathcal{J}^+$, all correct processes exchange $O(n)$ words (by \Cref{lemma:words_in_bigger_iteration}).
    Moreover, all correct processes acquire a certificate in the first iteration of $\mathcal{J}^+$ (by \Cref{lemma:liveness_helper}).
    Therefore, in all other iterations of $\mathcal{J}^+$, all correct processes send $O(f)$ words (by \Cref{lemma:silence}).
    Hence, all correct processes send $f \cdot O(n) + O(n) + O(n) \cdot O(f) = O(n \cdot f)$ words.
    The statement of the lemma holds in this case.

    \item Let $f > t_o$; hence, $f \in \Theta(n)$.
    In this case, \cert exchanges $O(n) \cdot O(n)$ words as each iteration of the for loop sends $O(n)$ words (by \Cref{lemma:words_in_bigger_iteration}).
    As $f \in \Theta(n)$, the statement of the lemma holds in this case as well.
\end{compactitem}
As the lemma holds in both possible scenarios, the proof is concluded.
\end{proof}
\section{Concluding Remarks} \label{section:conclusion}

This paper presents \name, the first synchronous SBA protocol with adaptive $O(n \cdot f)$ word complexity.
\name operates among $n = \big( 2 + \Omega(1) \big)t + 1$ processes while tolerating $t$ Byzantine failures.
In the future, we aim to adapt \name to achieve both adaptive $O(n \cdot f)$ word complexity and adaptive $O(f)$ latency. 




\bibliographystyle{acm}
\bibliography{references}


\end{document}